\documentclass{llncs}

\usepackage{epic}
\usepackage{eepic}
\usepackage{amsmath}
\usepackage{amssymb}
\usepackage{color}
\usepackage{latexsym}
\usepackage{graphics}
\usepackage{exscale}
\usepackage{epsfig}
\usepackage{rotate}
\usepackage[latin1]{inputenc}

\newtheorem{fact}{Fact}

\newcommand{\np}{\mbox{\sf NP}}

\newcommand{\hide}[1]{}

\title{Bounded-Distance Network Creation Games\thanks{This work was
partially supported by the PRIN 2008 research project COGENT
(COmputational and GamE-theoretic aspects of uncoordinated
NeTworks), funded by the Italian Ministry of Education,
University, and Research.}}
\author{Davide Bil\`o\inst{1} \and
Luciano Gual\`a\inst{2} \and Guido Proietti\inst{3,4}
\institute{Dip.to di Teorie e Ricerche dei Sistemi Culturali,
University of Sassari, Italy \and Dipartimento di Matematica,
University of Rome ``Tor Vergata'',  Italy \and Dipartimento di
Informatica, University of
  L'Aquila, Italy  \and Istituto di Analisi dei Sistemi
  ed Informatica,
  CNR, Rome, Italy \\
E-mail: {\tt davide.bilo@uniss.it; guala@mat.uniroma2.it;
guido.proietti@univaq.it} }}

\begin{document}

\pagestyle{plain}

\maketitle

\begin{abstract}
A \emph{network creation game} simulates a decentralized and non-cooperative building of a communication network. Informally, there are $n$ players 
sitting on the network nodes, which attempt to establish a reciprocal communication by activating, incurring a certain cost, 
any of their incident links. The goal of each player is to have all the other nodes as close as possible in the resulting network, while buying as few links as possible. According to this intuition, any model of the game must then appropriately address a balance between these two conflicting objectives. Motivated by the fact that a player might have a strong requirement about its centrality in the network, in this paper we introduce a new setting in which if a player maintains its (either \emph{maximum} or \emph{average}) distance to the other nodes within a given associated \emph{bound}, then its cost is simply equal to the \emph{number} of activated edges, otherwise its cost is unbounded.
We study the problem of understanding the structure of associated pure Nash equilibria of the resulting games, that we call \textsc{MaxBD} and
\textsc{SumBD}, respectively. 
For both games, we show that computing the best response of a player is an \np-hard problem.
Next, we show that when distance bounds associated with
players are \emph{non-uniform}, then equilibria can be arbitrarily bad.
On the other hand, for \textsc{MaxBD}, we show that when nodes have a \emph{uniform} bound $R$ on the maximum distance, then the \emph{Price of Anarchy} (PoA)
is lower and upper bounded by $2$ and $O\left(n^{\frac{1}{\lfloor\log_3
R\rfloor+1}}\right)$ for $R \ge 3$ (i.e., the PoA is constant as soon as the bound on the maximum distance is $\Omega(n^{\epsilon})$, for some $\epsilon>0$), while for the interesting case $R=2$, we are able to prove that the PoA is $\Omega(\sqrt{n})$ and $O(\sqrt{n \log n} )$. For the uniform \textsc{SumBD} we obtain similar (asymptotically) results, and moreover we show that the PoA becomes constant as soon as the bound on the average distance is $n^{\omega\left(\frac{1}{\sqrt{\log n}}\right)}$.
\end{abstract}
\smallskip

\newpage
\setcounter{page}{1}
\section{Introduction}
Communication networks are rapidly evolving towards a model in
which the constituting components (e.g., routers and links) are
activated and maintained by different owners, which one can
imagine as players sitting on the network nodes. When these
players act in a selfish way with the final intent of creating a
connected network, the challenge is exactly to understand whether
the pursuit of individual profit is compatible with the attainment
of an equilibrium status for the system (i.e., a status in which
players are not willing to move from), and how the social utility
for the system as a whole is affected by the selfish behavior of
the players. This task, which involves both computational and
economical issues of the system, is exactly the aim of a research
line which started with the seminal paper of Fabrikant \emph{et
al.} \cite{FLM03}, where the by now classic \emph{network creation
game} (NCG) was initially formalized and investigated.

\paragraph{Definition of the NCG.}
In its original formulation, the NCG is defined as follows: We are
given a set of $n$ players, say $V$, where the
strategy space of player $v \in V$ is the power set $2^{V \setminus\{v\}}$. Given
a combination of strategies $S=(S_v)_{v \in V}$, let $G(S)$ denote the underlying undirected
graph whose node set is $V$, and whose edge set is
$E(S)=\{\cup_{v \in V} (v \times S_v)\}$. Then, the \emph{cost} incurred by player
$v$ under $S$ is

\begin{equation}
\label{cost} \mathit{cost}_v(S) = \alpha \cdot |S_v| + \sum_{u \in V}
d_{G(S)}(u,v)
\end{equation}

\noindent where $d_{G(S)}(u,v)$ is the distance between nodes $u$
and $v$ in $G(S)$. Thus, the cost function implements the inherently antagonistic
goals of a player, which on the one hand
attempts to buy as little edges as possible, and on the other hand
aims to be as close as possible to the other nodes in the
outcoming network. These two criteria are suitably balanced in
(\ref{cost}) by making use of the parameter $\alpha \geq 0$.
Consequently, the \emph{Nash Equilibria}\footnote{In this paper, we only
focus on \emph{pure} strategies Nash equilibria.} (NE) space of the game is
heavily influenced by $\alpha$, and the corresponding
characterization must be given as a function of it. The state-of-the-art for the \emph{Price of Anarchy}
(PoA) of the game, that we will call henceforth \textsc{Sum}NCG,
is summarized in \cite{MS10}, where the most recent progresses on the problem have been reported.

\paragraph{Further NCG models.}
A criticism made to the classic NCG model is that the parameter $\alpha$ is in a sense exogenous to the system.
Moreover, usage and building cost are summed up together in the player's cost, and this mixing is reflected in the social cost
of the resulting network. As a consequence, we have that in this game the PoA alone does not say so much about the structural properties of the
network, such as density, diameter, or routing cost. This gave rise to a sequence of new NCG models. A first natural variant of \textsc{Sum}NCG was introduced in
\cite{DHM07}, where the authors redefined the player cost function
as follows

\begin{equation}
\label{max} \mathit{cost}_v(S) = \alpha \cdot |S_v| + \max \{d_{G(S)}(u,v):u \in V\}.
\end{equation}

\noindent This variant, named \textsc{Max}NCG, received further
attention in \cite{MS10}, where the authors improved the PoA of
the game on the whole range of values of $\alpha$. However, \textsc{Max}NCG still incorporates in its definition
the parameter $\alpha$. In an effort of defining new
parameter-free models, in \cite{LPR08} the authors proposed an
interesting variant in which a player $v$, when forming the
network, has a limited \emph{budget} $b_v$ to establish links to other
players. This way, the player cost function restricts to the usage cost, namely either the maximum or the
total distance to other nodes. In particular, in \cite{LPR08} the
authors focused on the latter measure. For this \emph{bounded-budget}
version of the game, that we call \textsc{Sum}BB, they showed that
determining the existence of NE is \np-hard. On a positive side,
they proved that for uniform budgets, say $k$, \textsc{Sum}BB
always admits a NE, and that its \emph{Price of Stability} (PoS)
is 1, while its PoA is $\Omega\left(\sqrt{\frac{n/k}{\log_k
n}}\right)$ and $O\left(\sqrt{\frac{n}{\log_k n}}\right)$. Notice
that in \textsc{Sum}BB, links are seen as directed. Thus, a
natural extension of the model was given in \cite{EFM11}, were the
undirected case was considered. For this, it was proven that both
\textsc{Max}BB and \textsc{Sum}BB always admit a NE. Moreover,
the authors showed that the PoA for \textsc{Max}BB and \textsc{Sum}BB is
$\Omega(\sqrt{\log n})$ and $2^{O(\sqrt{\log
n})}$, respectively, while in
the special case in which the budget is equal to 1 for all the
players, the PoA is $O(1)$ for both versions of the game.

In all the above models it must be noticed that, as stated in
\cite{FLM03}, for a player it is \np-hard to find a best response
once that the other players' strategies are fixed. To circumvent
this problem, in \cite{ADH10} the authors proposed a further
variant, called \emph{basic NCG} (BNCG), in which given some
existing network, the only improving transformations allowed are
\emph{edge swaps}, i.e., a player can only modify a \emph{single}
incident edge, by either replacing it with a new incident edge, or
by removing it. This naturally induces a weaker concept of
equilibrium for which a best response of a player can be computed
in polynomial time. In this setting, the authors were able to
give, among other results, an upper bound of $2^{O(\sqrt{\log
n})}$ for the PoA of \textsc{Sum}BNCG, and a lower bound of
$\Omega(\sqrt{n})$ for the PoA of \textsc{Max}BNCG. However,
as pointed out in \cite{MS10}, the fact that now an edge has not a
specific owner, prevents the possibility to establish any
implications on the PoA of the classic NCG, since a NE in a BNCG is not necessarily a NE of a NCG.

Finally, another NCG model which is barely related to the NCG model we study in this paper has been addressed in \cite{BS08}.

\paragraph{Our results.}
In this paper, we propose a new NCG variant that complements the
model proposed in \cite{EFM11}. More precisely, we assume that the
cost function of each player only consists of the number of bought
edges (without any budget on them), but with the additional constraint that a player $v$ needs to
connect to the network by staying within a given (either \emph{maximum} or \emph{average})
\emph{distance}, say (either $R_v$ or $D_v$), to the set of
players. Our model is motivated by the fact that in a realistic
scenario, a player might have a strong objective about its
centrality in the created network, and this can only be guaranteed by means of our approach.

For this bounded-distance version of the NCG, we address the problem of understanding the
structure of the NE associated with the two variants of the game,
that we denote by \textsc{MaxBD} and \textsc{SumBD}. To this
respect, we first show that both games can have an unbounded PoA
as soon as players hold at least two different distance bounds.
Moreover, in both games, computing a best response for a player is
\np-hard. These bad news are counterbalanced by the positive
result we get for \emph{uniform} distance bounds. In this case,
first of all, the PoS for \textsc{MaxBD} is equal to 1, while for
\textsc{SumBD} is at most equal to 2. Then, as far as the PoA is
concerned, let $R$ and $D$ denote the uniform bound on the maximum
and the average distance, respectively. We show that

\begin{enumerate}
\item[(i)] for \textsc{MaxBD}, the PoA is lower
and upper bounded by $2$ and $O\left(n^{\frac{1}{\lfloor\log_3
R\rfloor+1}}\right)$ for $R \ge 3$, while for $R=2$ is $\Omega(\sqrt{n})$ and $O(\sqrt{n \log n} )$; thus, the PoA is constant as soon as  $R=\Omega(n^{\epsilon})$, for some $\epsilon>0$;

\item[(ii)] for \textsc{SumBD}, the PoA is lower bounded by $2-\epsilon$, for any
$\epsilon > 0$, as soon as $D \geq 2-3/n$, while it is upper
bounded as reported in Table \ref{UBsum}.

\begin{table}
\begin{center}
\begin{tabular}{|l||c|c|c|c|c|}
\hline
{$D$}& $\in [2,3)$ & $\ge 3$ and $O(1)$ & $\omega(1) \cap O\left(3^{\sqrt{\log n}}\right)$ & $\omega\left(3^{\sqrt{\log n}}\right) \cap n^{O\left(\frac{1}{\sqrt{\log n}}\right)}$ & $n^{\omega\left(\frac{1}{\sqrt{\log n}}\right)}$\\
\hline
PoA&$O\left(\sqrt{n \log {n}}
\right)$ & $O\left(n^{\frac{1}{\lfloor\log_3
D/4\rfloor+2}}\right)$ &$2^{O\left(\sqrt{\log n}\right)}$ &$O\left(n^{\frac{1}{\lfloor\log_3
D/4\rfloor+2}}\right)$ & $O(1)$\\
\hline
\end{tabular}
\caption{Obtained PoA upper bounds for \textsc{SumBD}.}
\label{UBsum}
\end{center}
\end{table}
\end{enumerate}

The paper is organized as follows. After giving some basic
definitions in Section~\ref{sct:definitions}, we provide some preliminary results in Section~\ref{sct:preliminary}.
Then, we study upper and lower bounds for \textsc{MaxBD} and \textsc{SumBD} in Sections~\ref{sct:max} and~\ref{sct:sum}, respectively. Finally, in Section~\ref{sct:conclusions} we conclude the paper by discussing some intriguing relationships of our games with the famous graph-theoretic \emph{degree-diameter} problem. 

\section{Problem Definition}\label{sct:definitions}
\paragraph{Graph terminology.} Let $G=(V,E)$ be an
undirected (simple) graph with $n$ vertices. For a graph $G$, we
will also denote by $V(G)$ and $E(G)$ its set of vertices and its
set of edges, respectively. For every vertex $v \in V$, let
$N_G(v):=\{u \mid u\in V\setminus\{v\}, (u,v)\in E\}$. The {\em
minimum degree} of $G$ is equal to $\min_{v \in V}|N_G(v)|$.

We denote by $d_G(u,v)$ the {\em distance} in $G$ from $u$ to $v$.
The \emph{eccentricity} of a vertex $v$ in $G$, denoted by
$\varepsilon_G(v)$, is equal to $\max_{u \in V} d_G(u,v)$. The
{\em diameter} and the \emph{radius} of $G$ are equal to the
maximum and the minimum eccentricity of its nodes, respectively. A
node is said to be a \emph{center} of $G$ if $\varepsilon_G(v)$ is
equal to the radius of $G$. We define the \emph{broadcast cost} of
$v$ in $G$ as $B_G(v)=\sum_{u \in V}d_G(u,v)$, while the
\emph{average distance} from $v$ to a node in $G$ is denoted by
$D_G(v)=B_G(v)/n$.

A \emph{dominating set} of $G$ is a subset of nodes $U \subseteq
V$ such that every node of $V \setminus U$ is adjacent to some
node of $U$. We denote by $\gamma(G)$ the cardinality of a minimum
cardinality dominating set of $G$. Moreover, for any real $k \ge
1$, the $k$th power of $G$ is defined as the graph
$G^k=(V,E(G^k))$ where $E(G^k)$ contains an edge $(u,v)$ if and
only if $d_G(u,v) \le k$.

Let $U\subseteq V$ be a set of vertices, we denote by $G[U]$ the
subgraph of $G$ induced by $U$. Let $F\subseteq \{(u,v)\mid u,v \in V, u\neq v\}$.
We denote by $G+F$ (resp., $G-F$) the graph on $V$ with edge set
$E\cup F$ (resp., $E\setminus F$). When $F=\{e\}$ we will denote
$G+\{e\}$ (resp., $G-\{e\}$) by $G+e$ (resp., $G-e$). For two
graphs $G_1$ and $G_2$, we denote by $G_1\cup G_2$ the graph with
$V(G_1\cup G_2)=V(G_1)\cup V(G_2)$, and $E(G_1\cup G_2)=E(G_1)\cup
E(G_2)$.

\paragraph{Problem statements.}
The \emph{bounded maximum distance} NCG (\textsc{MaxBD}) is defined as
follows: Let $V$ be a set of $n$ nodes, each representing a
selfish player, and for any $v \in V$, let $R_v >0$ be an integer representing a bound on the eccentricity of $v$. The strategy of a player $v$ consists of a subset
$S_v \subseteq V \setminus \{v\}$. Denoting by $S$ the strategy profile of all players, let $G(S)$ be the undirected graph with
node set $V$, and with edge set $E(S)=\{\cup_{v \in V} (v \times S_v)\}$. When $u \in S_v$, we will say that $v$ is
buying the edge $(u,v)$, or that the edge $(u,v)$ is bought by $v$. Then, the cost of a player $v$ in $S$ is $\mathit{cost}_v(S)=|S_v|$
if $\varepsilon_{G(S)}(v) \le R_v$, $+\infty$ otherwise.

The \emph{bounded average distance} NCG (\textsc{SumBD}) is defined analogously, with a bound $D_v$ on the average distance, and cost function $\mathit{cost}_v(S)=|S_v|$ if $D_{G(S)}(v) \le D_v$, $+\infty$ otherwise. In the rest of the paper, depending on the context, we will interchangeably make use of the bound on the broadcast cost $B_v=D_v \cdot n$ when referring to \textsc{SumBD}.

In both variants, we say that a node $v$ is
\emph{within the bound in $S$} if $\mathit{cost}_v(S)<+\infty$. We measure the
overall quality of a graph $G(S)$ by its social cost
$\mathit{SC}(S)=\sum_{v \in V} \mathit{cost}_v(S)$. A graph $G(S)$ minimizing
$\mathit{SC}(S)$ is called \emph{social optimum}.

We use the Nash Equilibrium (NE) as solution concept. More
precisely, a NE is a strategy profile $S$ in which no player can
decrease its cost by changing its strategy assuming that the
strategies of the other players are fixed. When $S$ is a NE, we
will say that $G(S)$ is \emph{stable}, and that a graph $G$ is stable if
there exists a strategy profile $S$ such that $G=G(S)$. Notice that in
both games, when $S$ is a NE, all nodes are within the bound and, since
every edge is bought by a single player, $\mathit{SC}(S)$ coincides with
the number of edges of $G(S)$.

We conclude this section by recalling the definition of the two
measures we will use to characterize the NE space of our games,
namely the \emph{Price of Anarchy} (PoA) \cite{FLM03} and the
\emph{Price of Stability} (PoS) \cite{ADT03}, which are defined as
the ratio between the highest (respectively, the lowest) social
cost of a NE and the cost of a social optimum.

\section{Preliminary results}
\label{sct:preliminary}
First of all, observe that for \textsc{MaxBD} it is easy to see
that a stable graph always exists. Indeed, if there is at least
one node having distance bound 1, then the graph where all 1-bound
nodes buy edges towards all the other nodes is stable. Otherwise,
any spanning star is stable. Notice that any spanning star is stable for
\textsc{SumBD} as well, but only when all vertices have a bound of at least
$2n-3$, while the problem of understanding whether a NE always exists for the remaining values is open.
From these observations, we can derive the following negative
result:


\begin{theorem}
The PoA of \textsc{MaxBD} and \textsc{SumBD} (with distance bounds
$B_v \ge 2n-3$) is $\Omega(n)$, even for only two distance-bound
values.
\end{theorem}
\begin{proof}
We will define a graph $G$ with $\Omega(n^2)$ edges, and we will prove
that $G$ is stable for both versions of the game. Then,
we will show that in both cases the cost of the social optimum is
$n-1$.

The graph $G$ is defined as follows. We have a clique of $k$
nodes. For each node $v$ of the clique, we add four nodes
$v^1_1,v^1_2,v^2_1,v^2_2$ and four edges $(v^1_2,v^1_1),
(v^1_1,v),(v^2_2,v^2_1), (v^2_1,v)$. Clearly, $G$ has $n=5k$ nodes
and $\Omega(n^2)$ edges. Now, consider a strategy profile $S$ with
$G=G(S)$ and such that (i) every edge is bought by a single
player, and (ii) the edges $(v^j_2,v^j_1), (v^j_1,v)$ are bought
by $v^j_2$ and $v^j_1$, respectively, $j=1,2$. Now, we show that
$S$ is a NE, once we have defined suitable bounds for the players.

Let us consider \textsc{MaxBD} first. We set the bound of every
node of the clique to 3, while all the other nodes have bound 5.
Trivially, all nodes are within the bound. Moreover, a node
$v^j_i$ is buying only one edge and, since the removal of such
edge disconnects the graph, $v^j_i$ cannot decreases its cost. Let
$v$ be a node of the clique, and assume that $v$ is buying $h$
edges in $S$. Let $S'$ be a strategy profile where $v$ switches
its strategy $S_v$ with $S'_v$ and such that $|S'_v|<h$. Since $h
\le n-1$, it must exist a vertex $u$ of the clique such that $u
\in S_v$ and $u,u^1_1,u^1_2,u^2_1,u^2_2 \notin S'_v$, from which
we have that $v$ cannot be within the bound in $S'$ since
$d_{G(S')}(v,u^2_2)
> 3$.

Concerning \textsc{SumBD}, we set the bound of each node $v$ of
the clique to $\sum_{u \in V}d_G(v,u)=11k-5 > 2n-3$, while we
assign to all the other nodes bound $n^2$. Similar arguments used
for \textsc{MaxBD} can be used to show that $S$ is a NE for
\textsc{SumBD} as well.

To conclude the proof, observe that any star (with cost $n-1$) is
a social optimum for the two instances of \textsc{MaxBD} and \textsc{SumBD} given above. \qed
\end{proof}

Given the above bad news, from now on we focus our attention on the \emph{uniform} case of the games, i.e., that in which all the bounds on the distances are the same, say $R$ and $D$ (i.e., $B=D \cdot n$) for the maximum and the average version, respectively. Similarly to other NCGs, also here we have the problem of computing a best response for a player, as stated in the following theorem.

\begin{theorem}
\label{th:best response}
Computing the best response of a player in \textsc{MaxBD} and \textsc{SumBD} is \np-hard.
\end{theorem}
\begin{proof}
Let us consider \textsc{MaxBD} first. The reduction is from the \np-hard \emph{minimum
dominating set} problem which, given a graph $G'=(V',E')$, asks for
finding a dominating set of $G'$ of minimum cardinality, say
$\gamma(G')$. Let $N=|V'|$. We build a graph $G$ with $n=N+2N(R-2)
+ 1$ vertices as follows: We have an isolated vertex $u$, a copy
of $G'$, and two paths of
length $R-2$ appended to every vertex $v \in V'$. Now, let $S$ be the strategy profile such that
$G=G(S)$. Clearly, $\mathit{cost}_u(S)=+\infty$, and it is easy to see that
$u$ has a strategy yielding a cost of $k$ if and only if
$\gamma(G') \le k$.

Now, for \textsc{SumBD}, we sketch a reduction from the
\emph{$k$-median} problem. Let $G'=(V',E')$ be an instance of the
$k$-median problem which, given a value $\beta$, asks for finding
a subset $U \subseteq V$ of size $k$ such that $\sum_{v \in V}
\min_{u \in U}d_G(u,v) \le \beta$. This problem is \np-hard
even when $G'$ is an unweighted graph \cite{KH79}. Let $G$ be the
graph defined as $G'$ with an additional isolated node $u$, and
let $S$ be a strategy profile such that $G=G(S)$, and let
$B=\beta+N$, where $N=|V'|$. It is easy to see that $u$ has a
strategy yielding a cost of $k$ if and only if $G'$ has a
$k$-median of cost at most $\beta$. 
\qed
\end{proof}

On the other hand, a positive result which clearly implies that SUMBD always admits a pure NE is the following.

\begin{theorem}
\label{th:PoS}
The PoS of \textsc{MaxBD} is 1, while for \textsc{SumBD} is at most 2.
\end{theorem}
\begin{proof}
Concerning \textsc{MaxBD}, when $R=1$ the complete graph is a
social optimum as well as the only stable graph. For $R>1$, let
$T$ be a spanning star with center $c \in V$ and edges $(c,v)$, $v
\in V\setminus \{c\}$. Clearly, $T$ is a social optimum, and the
strategy profile $S$ in which $S_v=\{c\}$ for every $v \in
V\setminus \{c\}$, and $S_c=\emptyset$, is a NE.

Concerning \textsc{SumBD}, observe that such a $T$
is an optimum as well as stable when $B
\ge 2n-3$. Now assume that $B = n-1+k$ with $0 \le k \le n-2$. We
will define a graph $G$ with a number of edges that is at most
twice the number of edges of the optimum, and we show that it is
stable. Let $h,t \geq 0$ be s.t. $n=(k+1)h+t$. We partition $V$
into $h$ groups of $k+1$ nodes, say $V_1,\dots,V_h$ and, when $t
\neq 0$, an additional group $V_0$ of $t$ vertices. The edge set
of $G$ is defined as $\{(u,v) \mid u \in V_i, v \in V_j, i \neq j
\}$. Let $S$ be a strategy profile such that $G=G(S)$ with the
constraint that every node in $V_0$ buys no edge in $S$. Clearly,
every node in $G$ is within the bound. Moreover, observe that in
order to be within the bound, each node $v$ must have degree at
least $n-1-k$. Now, since every node not in $V_0$ has degree
exactly $n-1-k$, and since nodes in $V_0$ buy no edges, then
$G(S)$ is stable. To bound the social cost of $G$, notice that the
cost of the optimum, say \textsc{Opt}, is at least
$\frac{n(n-1-k)}{2}$. Let us consider the case in which $k < n/2$.
Then, $\textsc{Opt} \ge n^2/4$, while $\mathit{SC}(S) \le n^2/2
\le 2 \cdot \textsc{Opt}$. On the other hand, when $k \ge n/2$, we
have only two groups, one with $t=n-k-1$ nodes and the other with
$n-t$. Then, we have $\mathit{SC}(S)=t(n-t) \le t \, n \le 2 \cdot
\textsc{Opt}$. 
\qed
\end{proof}

We conclude this section by providing a lemma which will simplify the exposition of the remaining results.

\begin{lemma}\label{lm:aux}
Let $G(S)$ be a stable graph and let $H$ be a subgraph of $G(S)$.
If for each node $v$ there exists a set $E_v$ of edges (all
incident to $v$) such that $v$ is within the bound in $H+E_v$, then
$\mathit{SC}(S) \le |E(H)|+ \sum_{v \in V}|E_v|$.
\end{lemma}
\begin{proof}
Let $k_v$ be the number of edges of $H$ that $v$ is buying in $S$.
If $v$ buys $E_v$ additionally to its $k_v$ edges, then $v$ will
be within the bound. Hence, since $S$ is a NE, we have that
$\mathit{cost}_v(S) \le k_v+|E_v|$, from which it follows that:
$$
\mathit{SC}(S)=\sum_{v \in V}\mathit{cost}_v(S) \le \sum_{v \in V}k_v + \sum_{v \in
V}|E_v|= |E(H)|+\sum_{v \in V} |E_v|.
$$ \qed
\end{proof}

\section{PoA for \textsc{MaxBD}}
\label{sct:max}
\subsection{Upper bounds}
\begin{lemma}\label{lm:PoA le gamma}
Let $G(S)$ be a NE, and let $\gamma$ be the cardinality of a
minimum dominating set of $G(S)^{R-1}$, then $\mathit{SC}(S) \le
(\gamma+1)(n-1)$.
\end{lemma}
\begin{proof}
Let $U$ be a minimum dominating set of $G(S)^{R-1}$, with
$\gamma=|U|$. It is easy to see that there is a spanning forest
$F$ of $G(S)$ consisting of $\gamma$ trees $T_1, \dots,
T_{\gamma}$, such that every $T_j$ contains exactly one vertex in
$U$, and when we root $T_j$ at such vertex the height of $T_j$ is
at most $R-1$.

For a node $v \in V$, let $E_v=\{(v,u) \mid u \in
U\setminus\{v\}\}$. Clearly, $v$ is within the bound in $F+E_v$, hence by
using Lemma \ref{lm:aux}, we have
$$
\mathit{SC}(S) \le |E(F)|+\sum_{u \in U}|E_u|+\sum_{v \in V\setminus
U}|E_v| = n-\gamma +(\gamma-1) \gamma+
\gamma(n-\gamma)\le(\gamma+1)(n-1).
$$ \qed
\end{proof}

Let $G(S)$ be a NE and let $v$ be a node of $G(S)$. 
Since $v$ is within the bound, the neigborhood of $v$ in $G$ is a dominating
set of $G^{R-1}$. Therefore, thanks to Lemma \ref{lm:PoA le gamma} we have proved the following corollary.

\begin{corollary}\label{lm:PoA le delta}
Let $G(S)$ be a NE, and let $\delta$ be the minimum degree of
$G(S)$, then $\mathit{SC}(S) \le (\delta+1)(n-1)$.
\qed
\end{corollary}

We are now ready to prove our upper bound to the PoA for the game.

\begin{theorem}\label{th:UB PoA Max}
The PoA of \textsc{MaxBD} is $O(n^{\frac{1}{\lfloor\log_3
R\rfloor+1}})$ for $R \ge 3$, and $O( \sqrt{n \log n} )$ for $R=2$.
\end{theorem}
\begin{proof}
Let $G$ be a stable graph, and let $\gamma$ be the size of a
minimum dominating set of $G^{R-1}$. We define the \emph{ball} of radius
$k$ centered at a node $u$ as $\beta_k(u)=\{v \mid d_G(u,v) \le
k\}$. Moreover, let $\beta_k=\min_{u \in V} |\beta_k(u)|$. The
idea is to show that in $G$ the size of any ball increases quite
fast as the radius of the ball increases.
\begin{claim}
For any $k \ge 1$, we have $\beta_{3k+1} \ge \min \{n,\gamma
\beta_k\}$.
\end{claim}
\begin{proof}
Consider the ball $\beta_{3k+1}(u)$ centered at any given node
$u$, and assume that $|\beta_{3k+1}(u)| \le n$. Let $T$ be the
maximal set of nodes at distance exactly $2k+1$ from $u$ and
subject to the distance between any pair of nodes in $T$ being at
least $2k+1$. We claim that for every node $v \notin
\beta_{3k+1}(u)$, there is a vertex $t \in T$ with
$d_G(t,v)<d_G(u,v)$. Indeed, consider the node $t'$ in the
shortest path between $v$ and $u$ at distance exactly $2k+1$ from
$u$. If $t' \in T$ the claim trivially holds, otherwise consider
the node $t \in T$ that is closest to $t'$. From the maximality of
$T$ we have that $d_G(t,v) \le d_G(t,t')+d_G(t',v) \le
2k+d_G(u,v)-(2k+1) < d_G(u,v)$.

As a consequence, we have that $T\cup\{u\}$ is a dominating set of
$G^{R-1}$, and hence $|T|+1 \ge \gamma$. Moreover, all the balls
centered at nodes in $T\cup\{u\}$ with radius $k$ are all pairwise
disjoint. Then:
$$
|\beta_{3k+1}(u)|\ge |\beta_k(u)|+\sum_{t \in T} |\beta_k(t)| \ge
\gamma \beta_k.
$$
\qed
\end{proof}

Now, observe that since the neighborhood of any node is a
dominating set of $G^{R-1}$, we have that $\beta_1 \ge \gamma$.
Then, after using the above claim $x$ times, we obtain
$$
\beta_{\frac{3^{x+1}-1}{2}} \ge \min \{n,\gamma^{x+1}\}.
$$

Let us consider the case $R\ge 3$ first. Let $U$ be a maximal
independent set of $G^{R-1}$. Since $U$ is also a dominating set
of $G^{R-1}$, it holds that $|U| \ge \gamma$. We consider the
$|U|$ balls centered at nodes in $U$ with radius given by the
value of the parameter $x=\lfloor\log_3 R-1\rfloor$. Every ball
has radius at most $(R-1)/2$ and since $U$ is an independent set
of $G^{R-1}$, all balls are pairwise disjoint and hence we have $n
\ge |U| \gamma^{\lfloor\log_3 R-1\rfloor + 1} \ge
\gamma^{\lfloor\log_3 R\rfloor+1}$. As a consequence, we obtain
$\gamma \le n^{\frac{1}{\lfloor\log_3 R\rfloor+1}}$, and the claim
now follows from Lemma \ref{lm:PoA le gamma}.

Now assume $R=2$. We use the bound given in \cite{AS92} to the
size $\gamma(G)$ of a minimum dominating set of a graph $G$ with
$n$ nodes and minimum degree $\delta$, namely $\gamma(G) \le
\frac{n}{\delta +1}H_{\delta+1}$, where $H_i=\sum_{j=1}^i 1/j$ is
the $i$-th harmonic number. Hence, since a social optimum has cost
$n-1$, from Lemma \ref{lm:PoA le gamma} and Corollary \ref{lm:PoA
le delta}, we have $\frac{\mathit{SC}(S)}{n-1} \le \min\{\delta +
1, \frac{n}{\delta+1} H_{\delta+1}+1 \}$, for any stable graph
$G(S)$ with minimum degree $\delta$. The claim follows. \qed
\end{proof}

\subsection{Lower bounds}
We first prove a simple constant lower bound for any value of
$R=o(n)$, and then we show an almost tight lower bound of $\Omega(\sqrt{n})$ for $R=2$. We postpone to the concluding section a discussion on the difficulty of finding better lower bounds for large values of $R$.

\begin{theorem}
For any $\epsilon>0$ and for $1<R=o(n)$, the PoA for \textsc{MaxBD} is at least
$2-\epsilon$.
\end{theorem}
\begin{proof}
Assume we are given a set of $n=2R+h$ vertices $\{u_1,\dots,u_{2R}\}\cup\{v_1,\dots,v_h\}$. The strategy profile $S$ is defined as
follows. Vertex $u_j$ buys a single edge towards $u_{j+1}$, for
each $j=1,\dots,2R-1$, and every $v_i$ buys two edges towards
$u_1$ and $u_{2R}$. It is easy to see that $G(S)$ has diameter $R$
and is stable. The claim follows from the fact that $SC(S)$ goes
to $2(n-1)$ as $h$ goes to infinity and the fact that, as observed in Section~\ref{sct:preliminary}, a spanning star (having social cost equal to $n-1$) is a social optimum. \qed
\end{proof}

We close this section by providing a much stronger lower bound for
the special case in which $R=2$. Before stating the theorem, we give some additional notation.
Let $v$ be a player, $S$
be a strategy profile, and $\ell$ a positive integer. We define
$N_{S}^\ell(v)=\{u \mid u\in V, d_{G(S)}(u,v)\leq \ell\}$ and
$\bar N_{S}^\ell(v)=V\setminus N_S^\ell(v)$. We will omit the superscript $\ell$ when $\ell=1$. 
Moreover, we denote
by $\gamma(S,v,\ell)$ the size of a minimum cardinality set
$X\subseteq V$ of vertices that {\em dominates $\bar
N_{S}^\ell(v)$ in $G(S)^{\ell-1}$}, i.e., for every vertex $u \in
\bar N_{S}^\ell(v)$ there exists a vertex $x \in X$ such that
$d_{G(S)}(x,u)\leq \ell-1$, i.e., $(x,u)\in E(G(S)^{\ell-1})$.
Finally, denote by $S_{\neg v}$ the strategy profile where each player
but $v$ plays the same strategy as in $S$, while $v$ buys no edge,
i.e., the strategy of $v$ is $\emptyset$. The
following proposition, whose proof is straightforward, provides
exact bounds to the cost incurred by each player in every
connected graph, and will be used in the proof of the theorem.

\begin{proposition}\label{prop: bounds for S_v max game}
Let $G(S)$ be a connected graph. The cost incurred by each player
$v$ in $S$ in \textsc{MaxBD} with bound $R$ is
$|S_v|\geq\gamma(S_{\neg v},v,R)$. Moreover, a player $v$ is in equilibrium in $S$ iff 
$|S_v|=\gamma(S_{\neg v},v,R)$.
\end{proposition}

Let $S'$ be a strategy profile for a set of players $V$. A strategy profile $S$ for $V$ {\em extends} $S'$ 
if $S'_v\subseteq S_v$ for every $v \in V$. Let $S$ and $S'$ be two strategy profiles for a set of players $V$ such that $S$ extends $S'$, 
For every $v \in V$ let $N(S',S,v)=N_S(v)\setminus N_{S'}(v)$ and let $S^{v,S'}$ be the strategy profile 
such that $S^{v,S'}_v=S'_v\cup N(S',S,v)$ and $S^{v,S'}_u=S'_u\cup \{v' \mid v' \in S_u, v'\neq v\}$ for each $u \in V, u \neq v$. 
Observe that $G(S^{v,S'})=G(S)$.

The following proposition will be also used in the proof of the theorem.

\begin{proposition}\label{prop: sufficient conditions for NE}
Let $V$ be a set of players and let $S,S'$ be two strategy profiles for $V$ such that $S$ extends $S'$. 
If every player $v$ is in equilibium in $S^{v,S'}$, then $G(S)$ is a stable graph.
\end{proposition}
\begin{proof}
For the sake of contradiction, assume that every player $v \in V$ is in equilibium in $S^{v,S'}$ but $G(S)$ is not stable. Then there
exists a player $u$ and a strategy profile $S''$ such that (i) $S''_v=S_v$ for every $v \in V, v\neq u$, 
(ii) the eccentricity of $u$ in $G(S'')$ is less than or equal to $R$, and (iii) $|S''_u|<|S_u|$.
 Let $X=\{x\mid x \in S''_u, x \not\in S_u\}$ and let $Y=\{y\mid y \in S_u, y\not\in S''_u\}$.
By (iii) we have $|X|<|Y|$. Let $\bar S$ be the strategy profile such that $\bar S_u=(S^{u,S'}_u\setminus Y)\cup X$ and $\bar S_v=S^{u,S'}_v$ for every $v \in V,v\neq u$.
Clearly, $G(S'')=G(\bar S)$ and thus, by (ii) the eccentricity of $u$ in $G(\bar S)$ is less than or equal to 2. Furthermore, $|X|<|Y|$ implies $|\bar S_u|\leq |S^{u,S'}_u|$ and 
therefore $u$ is not in equilibrium in $S^{u,S'}$.
 \qed
\end{proof}

We are now ready to prove the following.

\begin{theorem}
\label{th:PoA R=2}
The PoA of \textsc{MaxBD} for $R=2$ is $\Omega(\sqrt{n})$.
\end{theorem}
\begin{proof}
Let $p\geq 3$ be a prime number. We provide a graph $G'$ of diameter 2 containing 
$O(p^2)$ vertices and $\Omega(p^3)$ edges and show that there exists a strategy profile $S$ 
such that $G(S)=G'$ and $G(S)$ is stable.
$G'$ contains two vertex-disjoint rooted trees
$T$ and $T'$ as subgraphs. $T$ is a complete $p$-ary tree of
height 2. We denote by $r$ the root of $T$, by
$C=\{c_0,\ldots,c_{p-1}\}$ the set of children of $r$, and by
$V_i=\{v_{i,0},\ldots,v_{i,p-1}\}$ the set of children of $c_i$.
$T'$ is a star with $p^2$ leaves rooted at the center $r'$. The
leaves of $T'$ are partitioned in $p$ groups  each having exactly
$p$ vertices. For every $i=0,\ldots,p-1$, we denote by
$U_i=\{u_{i,0},\ldots,u_{i,p-1}\}$ the set of vertices of group
$i$. $G'=(V,E)$ has vertex set $V=V(T)\cup V(T')$ and edge set (see also Figure \ref{fig:lower bound sqrt n})
\begin{eqnarray*}
E &=& E(T)\cup E(T')\cup\{(r,r')\}\\
&\cup &\big\{(c,c')\mid c,c'\in C, c \neq c'\big\}\\
&\cup& \bigcup_{i=0}^{p-1}\big\{(u,u')\mid u,u'\in U_i, u\neq u'\big\}\\
&\cup &\big\{(u_{i,j},v_{i',j'})\mid i,i',j,j'\in[p-1],
j+i'i\equiv j' \pmod{p}\big\}.
\end{eqnarray*}

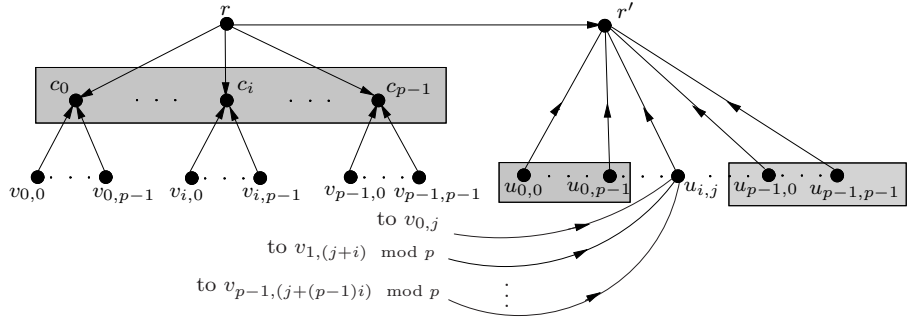
\begin{figure}
    \begin{center}
        \setlength{\unitlength}{0.00083333in}
\begingroup\makeatletter\ifx\SetFigFont\undefined%
\gdef\SetFigFont#1#2#3#4#5{%
  \reset@font\fontsize{#1}{#2pt}%
  \fontfamily{#3}\fontseries{#4}\fontshape{#5}%
  \selectfont}%
\fi\endgroup%
{\renewcommand{\dashlinestretch}{30}
\begin{picture}(5695,2013)(0,-10)
\put(2294,533){\makebox(0,0)[lb]{{\SetFigFont{8}{9.6}{\rmdefault}{\mddefault}{\updefault}to $v_{0,j}$}}}
\put(3065.283,2233.088){\arc{3481.909}{0.8811}{1.7366}}
\put(1170,105){\makebox(0,0)[lb]{{\SetFigFont{8}{9.6}{\rmdefault}{\mddefault}{\updefault}to $v_{p-1,(j+(p-1)i) \mod p}$}}}
\put(1611,346){\makebox(0,0)[lb]{{\SetFigFont{8}{9.6}{\rmdefault}{\mddefault}{\updefault}to $v_{1,(j+i) \mod p}$}}}
\put(3087,58){\makebox(0,0)[lb]{{\SetFigFont{8}{9.6}{\rmdefault}{\mddefault}{\updefault}$\vdots$}}}
\put(3180.275,1027.058){\arc{2037.367}{0.1657}{2.0145}}
\path(3487,546)(3614,577)
\blacken\path(3541.024,538.600)(3614.000,577.000)(3531.539,577.459)(3541.024,538.600)
\path(3595,103)(3684,149)
\blacken\path(3622.114,94.501)(3684.000,149.000)(3603.748,130.035)(3622.114,94.501)
\texture{0 115111 51000000 444444 44000000 151515 15000000 444444 
	44000000 511151 11000000 444444 44000000 151515 15000000 444444 
	44000000 115111 51000000 444444 44000000 151515 15000000 444444 
	44000000 511151 11000000 444444 44000000 151515 15000000 444444 }
\shade\path(4491,961)(5597,961)(5597,694)
	(4491,694)(4491,961)
\path(4491,961)(5597,961)(5597,694)
	(4491,694)(4491,961)
\path(3553,393)(3604,418)
\blacken\path(3540.969,364.829)(3604.000,418.000)(3523.363,400.746)(3540.969,364.829)
\texture{44000000 aaaaaa aa000000 8a888a 88000000 aaaaaa aa000000 888888 
	88000000 aaaaaa aa000000 8a8a8a 8a000000 aaaaaa aa000000 888888 
	88000000 aaaaaa aa000000 8a888a 88000000 aaaaaa aa000000 888888 
	88000000 aaaaaa aa000000 8a8a8a 8a000000 aaaaaa aa000000 888888 }
\shade\path(3057,963)(3874,963)(3874,715)
	(3057,715)(3057,963)
\path(3057,963)(3874,963)(3874,715)
	(3057,715)(3057,963)
\put(3095.500,1630.512){\arc{2635.084}{0.6079}{1.8416}}
\shade\path(167,1557)(2727,1557)(2727,1210)
	(167,1210)(167,1557)
\path(167,1557)(2727,1557)(2727,1210)
	(167,1210)(167,1557)
\put(4520,745){\makebox(0,0)[lb]{{\SetFigFont{8}{9.6}{\rmdefault}{\mddefault}{\updefault}$u_{p-1,0}$}}}
\put(3480,747){\makebox(0,0)[lb]{{\SetFigFont{8}{9.6}{\rmdefault}{\mddefault}{\updefault}$u_{0,p-1}$}}}
\put(4995,728){\makebox(0,0)[lb]{{\SetFigFont{8}{9.6}{\rmdefault}{\mddefault}{\updefault}$u_{p-1,p-1}$}}}
\put(4209,742){\makebox(0,0)[lb]{{\SetFigFont{8}{9.6}{\rmdefault}{\mddefault}{\updefault}$u_{i,j}$}}}
\put(3089,739){\makebox(0,0)[lb]{{\SetFigFont{8}{9.6}{\rmdefault}{\mddefault}{\updefault}$u_{0,0}$}}}
\put(4946,881){\blacken\ellipse{8}{8}}
\put(4946,881){\ellipse{8}{8}}
\put(5062,881){\blacken\ellipse{8}{8}}
\put(5062,881){\ellipse{8}{8}}
\put(4741,883){\blacken\ellipse{78}{78}}
\put(4741,883){\ellipse{78}{78}}
\put(5167,882){\blacken\ellipse{78}{78}}
\put(5167,882){\ellipse{78}{78}}
\put(263,870){\blacken\ellipse{8}{8}}
\put(263,870){\ellipse{8}{8}}
\put(384,870){\blacken\ellipse{8}{8}}
\put(384,870){\ellipse{8}{8}}
\put(500,870){\blacken\ellipse{8}{8}}
\put(500,870){\ellipse{8}{8}}
\put(179,872){\blacken\ellipse{78}{78}}
\put(179,872){\ellipse{78}{78}}
\put(605,871){\blacken\ellipse{78}{78}}
\put(605,871){\ellipse{78}{78}}
\put(1225,864){\blacken\ellipse{8}{8}}
\put(1225,864){\ellipse{8}{8}}
\put(1346,864){\blacken\ellipse{8}{8}}
\put(1346,864){\ellipse{8}{8}}
\put(1462,864){\blacken\ellipse{8}{8}}
\put(1462,864){\ellipse{8}{8}}
\put(1141,866){\blacken\ellipse{78}{78}}
\put(1141,866){\ellipse{78}{78}}
\put(1567,865){\blacken\ellipse{78}{78}}
\put(1567,865){\ellipse{78}{78}}
\put(2217,869){\blacken\ellipse{8}{8}}
\put(2217,869){\ellipse{8}{8}}
\put(2338,869){\blacken\ellipse{8}{8}}
\put(2338,869){\ellipse{8}{8}}
\put(2454,869){\blacken\ellipse{8}{8}}
\put(2454,869){\ellipse{8}{8}}
\put(2133,871){\blacken\ellipse{78}{78}}
\put(2133,871){\ellipse{78}{78}}
\put(2559,870){\blacken\ellipse{78}{78}}
\put(2559,870){\ellipse{78}{78}}
\put(4825,881){\blacken\ellipse{8}{8}}
\put(4825,881){\ellipse{8}{8}}
\put(4611,877){\blacken\ellipse{8}{8}}
\put(4611,877){\ellipse{8}{8}}
\put(4495,877){\blacken\ellipse{8}{8}}
\put(4495,877){\ellipse{8}{8}}
\put(2304,1350){\blacken\ellipse{78}{78}}
\put(2304,1350){\ellipse{78}{78}}
\put(1359,1350){\blacken\ellipse{78}{78}}
\put(1359,1350){\ellipse{78}{78}}
\put(417,1350){\blacken\ellipse{78}{78}}
\put(417,1350){\ellipse{78}{78}}
\put(1359,1829){\blacken\ellipse{78}{78}}
\put(1359,1829){\ellipse{78}{78}}
\put(3718,1816){\blacken\ellipse{78}{78}}
\put(3718,1816){\ellipse{78}{78}}
\put(3213,883){\blacken\ellipse{78}{78}}
\put(3213,883){\ellipse{78}{78}}
\put(4374,877){\blacken\ellipse{8}{8}}
\put(4374,877){\ellipse{8}{8}}
\path(1348,1856)(1352,1381)
\blacken\path(1331.327,1460.829)(1352.000,1381.000)(1371.326,1461.166)(1331.327,1460.829)
\path(1356,1826)(2273,1385)
\blacken\path(2192.236,1401.648)(2273.000,1385.000)(2209.572,1437.696)(2192.236,1401.648)
\path(182,910)(398,1318)
\blacken\path(378.245,1237.939)(398.000,1318.000)(342.893,1256.655)(378.245,1237.939)
\path(1130,907)(1346,1315)
\blacken\path(1326.245,1234.939)(1346.000,1315.000)(1290.893,1253.655)(1326.245,1234.939)
\path(2124,910)(2294,1322)
\blacken\path(2281.974,1240.420)(2294.000,1322.000)(2244.998,1255.677)(2281.974,1240.420)
\path(596,911)(422,1318)
\blacken\path(471.838,1252.302)(422.000,1318.000)(435.058,1236.578)(471.838,1252.302)
\path(1553,899)(1379,1306)
\blacken\path(1428.838,1240.302)(1379.000,1306.000)(1392.058,1224.578)(1428.838,1240.302)
\path(2548,912)(2323,1322)
\blacken\path(2379.021,1261.489)(2323.000,1322.000)(2343.954,1242.245)(2379.021,1261.489)
\path(3220,926)(3700,1789)
\path(3747,907)(3719,1789)
\path(4157,917)(3723,1799)
\path(4712,926)(3747,1803)
\path(5146,921)(3761,1808)
\drawline(3429,1295)(3429,1295)
\path(3415,1295)(3443,1332)
\blacken\path(3410.673,1256.139)(3443.000,1332.000)(3378.776,1280.276)(3410.673,1256.139)
\path(3733,1267)(3733,1327)
\blacken\path(3753.000,1247.000)(3733.000,1327.000)(3713.000,1247.000)(3753.000,1247.000)
\path(3994,1239)(3961,1299)
\blacken\path(4017.078,1238.541)(3961.000,1299.000)(3982.029,1219.264)(4017.078,1238.541)
\path(4321,1271)(4269,1313)
\blacken\path(4343.802,1278.292)(4269.000,1313.000)(4318.669,1247.174)(4343.802,1278.292)
\path(4540,1304)(4484,1341)
\blacken\path(4561.772,1313.586)(4484.000,1341.000)(4539.722,1280.213)(4561.772,1313.586)
\put(3606,878){\blacken\ellipse{8}{8}}
\put(3606,878){\ellipse{8}{8}}
\put(3490,878){\blacken\ellipse{8}{8}}
\put(3490,878){\ellipse{8}{8}}
\put(3369,878){\blacken\ellipse{8}{8}}
\put(3369,878){\ellipse{8}{8}}
\put(4175,877){\blacken\ellipse{78}{78}}
\put(4175,877){\ellipse{78}{78}}
\put(3749,878){\blacken\ellipse{78}{78}}
\put(3749,878){\ellipse{78}{78}}
\put(4070,876){\blacken\ellipse{8}{8}}
\put(4070,876){\ellipse{8}{8}}
\path(1364,1822)(3689,1822)
\blacken\path(3609.000,1802.000)(3689.000,1822.000)(3609.000,1842.000)(3609.000,1802.000)
\put(1317,1886){\makebox(0,0)[lb]{{\SetFigFont{8}{9.6}{\rmdefault}{\mddefault}{\updefault}$r$}}}
\put(3793,1864){\makebox(0,0)[lb]{{\SetFigFont{8}{9.6}{\rmdefault}{\mddefault}{\updefault}$r'$}}}
\put(1423,1393){\makebox(0,0)[lb]{{\SetFigFont{8}{9.6}{\rmdefault}{\mddefault}{\updefault}$c_i$}}}
\put(522,709){\makebox(0,0)[lb]{{\SetFigFont{8}{9.6}{\rmdefault}{\mddefault}{\updefault}$v_{0,p-1}$}}}
\put(0,712){\makebox(0,0)[lb]{{\SetFigFont{8}{9.6}{\rmdefault}{\mddefault}{\updefault}$v_{0,0}$}}}
\put(1002,709){\makebox(0,0)[lb]{{\SetFigFont{8}{9.6}{\rmdefault}{\mddefault}{\updefault}$v_{i,0}$}}}
\put(1454,709){\makebox(0,0)[lb]{{\SetFigFont{8}{9.6}{\rmdefault}{\mddefault}{\updefault}$v_{i,p-1}$}}}
\put(2400,713){\makebox(0,0)[lb]{{\SetFigFont{8}{9.6}{\rmdefault}{\mddefault}{\updefault}$v_{p-1,p-1}$}}}
\put(1968,721){\makebox(0,0)[lb]{{\SetFigFont{8}{9.6}{\rmdefault}{\mddefault}{\updefault}$v_{p-1,0}$}}}
\put(3954,876){\blacken\ellipse{8}{8}}
\put(3954,876){\ellipse{8}{8}}
\put(3833,876){\blacken\ellipse{8}{8}}
\put(3833,876){\ellipse{8}{8}}
\put(1037,1352){\blacken\ellipse{8}{8}}
\put(1037,1352){\ellipse{8}{8}}
\put(921,1352){\blacken\ellipse{8}{8}}
\put(921,1352){\ellipse{8}{8}}
\put(254,1398){\makebox(0,0)[lb]{{\SetFigFont{8}{9.6}{\rmdefault}{\mddefault}{\updefault}$c_0$}}}
\put(2353,1368){\makebox(0,0)[lb]{{\SetFigFont{8}{9.6}{\rmdefault}{\mddefault}{\updefault}$c_{p-1}$}}}
\put(800,1352){\blacken\ellipse{8}{8}}
\put(800,1352){\ellipse{8}{8}}
\put(1995,1348){\blacken\ellipse{8}{8}}
\put(1995,1348){\ellipse{8}{8}}
\put(1879,1348){\blacken\ellipse{8}{8}}
\put(1879,1348){\ellipse{8}{8}}
\path(1356,1839)(440,1368)
\blacken\path(502.000,1422.369)(440.000,1368.000)(520.291,1386.796)(502.000,1422.369)
\put(1758,1348){\blacken\ellipse{8}{8}}
\put(1758,1348){\ellipse{8}{8}}
\end{picture}
}
    \end{center}
        \caption{The graph $G(S)$. Edges are bought from the nodes they exit from. Notice that nodes in grey boxes are clique-connected (with arbitrary orientations), and for the sake of readability we have only inserted edges leading to node $u_{i,j}$.}
\label{fig:lower bound sqrt n}
\end{figure}

We claim that the diameter of $G'$ is 2. The eccentricity
of $r$ is 2 as $T$ has height 2, $T'$ has height 1, and $G'$
contains the edge $(r,r')$. Observe that the subgraphs of $G'$
induced by $C$ and $U_i$, for all $i\in[p-1]$, are all cliques of
$p$ vertices. Furthermore, by construction, there is an edge
linking each vertex $u\in \bar U$ with some $v \in V_i$, for every
$i\in[p-1]$, and thus $V_i$ dominates $\bar U$. Therefore, the
eccentricity of each vertex in $C$ is 2. As a consequence, to
prove that $G'$ has diameter 2, it is enough to prove that
\begin{enumerate}
\item[(i)] $U_i$ dominates $\bar V$, for every $i\in[p-1]$ (so as
each vertex in $\bar U$ would have eccentricity 2), \item[(ii)]
for every pair $v \in V_i$ and $v'\in V_{i'}$, $i,i'\in[p-1],i\neq
i'$, there is a vertex $u \in \bar U$ such that $(v,u),(v',u)\in
E$ (so as each vertex in $\bar V$ would have eccentricity 2).
\end{enumerate}

To prove (i), simply observe that for every $i',j'\in[p-1]$, there
always exists a $j\in[p-1]$ such that $j+i'i\equiv j' \pmod{p}$,
and thus, $(u_{i,j},v_{i',j'}) \in E$. To prove (ii), observe that
for every $v_{i,j},v_{i',j'}\in \bar V$, with $i\neq i'$, there
always exists $i'',j''\in[p-1]$ such that $j''+ii''\equiv j
\pmod{p}$ and $j''+i'i''\equiv j' \pmod{p}$ as $p$ is a prime
number (simply choose $i''$ such that $i''(i-i')\equiv(j-j') \pmod
p$). Therefore, $(v_{i,j},u_{i'',j''}),(v_{i',j'},u_{i'',j''})\in
E$.

To complete the proof, it remains to show that there exists a strategy profile 
$S$ such that $G(S)=G'$ and $G(S)$ is stable.
Let $\bar V$ and $\bar U$ be the set of leaves of $T$ and $T'$,
respectively. Let $S'$ be a strategy profile where:
\begin{itemize}
\item each vertex in $\bar V\cup\{r\}$ buys all edges
incident to it (thus, each vertex in $\bar V\cup\{r\}$ buys
exactly $p+1$ edges), 
\item each vertex in $\bar U$ buys the edge towards $r'$,
\item each of the remaining vertices buys no edge, i.e., $S'_v=\emptyset$ for every $v \in C\cup\{r'\}$.
\end{itemize}

Let $S$ be any strategy profile that extends $S'$ such that $G(S)=G'$. Observe that 
\begin{itemize}
\item $r'$ buys no edge in $S^{r,S'}$,
\item each vertex $v$ in $C\cup \bar U$ buys exactly $p$ edges in $S^{v,S'}$,
\item each vertex $v$ in $\bar V\cup\{r\}$ buys exactly $p+1$ edges in $S^{v,S'}$.
\end{itemize}

First of all observe that
$$
|\bar N_{S^{v,S'}_{\neg v}}^2(v)\cap \bar V|=
\begin{cases}
0 & \text{if $v=r'$;}\\
p^2 & \text{if $v=r$;}\\
p^2-1 & \text{if $v \in \bar V$;}\\
p(p-1) & \text{otherwise.}
\end{cases}
$$

Let $\hat V\subseteq \bar V$ be such that $|\hat
V|\in\{p^2,p^2-1,p(p-1)\}$, and let $X$ be a set of vertices that
dominates $\hat V$ in $G(S)$.
\begin{claim}
$|X|\geq \lceil\frac{|\hat
V|}{p}\rceil$ where equality holds only if $X\subseteq C$ or $X\subseteq\bar U$.
\end{claim}
\begin{proof}
Let $X$ be a set of vertices that dominates $\hat V$ in $G(S)$ and 
observe that $X\subseteq \hat V \cup
C\cup\bar U$. As any vertex of $G(S)$ can dominate at most $p$
vertices of $\hat V$, we have that $|X|\geq \lceil\frac{|\hat
V|}{p}\rceil$. 

Now we prove that if $|X|=\lceil\frac{|\hat
V|}{p}\rceil$ then either $X\subseteq C$ or $X\subseteq\bar U$. Indeed,
any set $X'$ dominating $\hat V$ in $G(S)$ and containing any
vertex of $\hat V$ has size $|X'|\geq 1+\lceil\frac{|\hat
V|-1}{p}\rceil >\lceil\frac{|\hat V|}{p}\rceil$, where the first
inequality holds because  any vertex in $\hat V$ dominates only
itself while the second inequality holds by the choice of $|\hat
V|$ and because $p\geq 3$. Furthermore, any set $X'$ containing $0<k<\lceil\frac{|\hat
V|}{p}\rceil$ vertices of $C$ and $\lceil\frac{|\hat
V|}{p}\rceil-k$ vertices of $\bar U$ can dominate at most
$pk+p(\lceil\frac{|\hat V|}{p}\rceil-k)-k(\lceil\frac{|\hat
V|}{p}\rceil-k)=p\lceil\frac{|\hat V|}{p}\rceil-k\lceil\frac{|\hat
V|}{p}\rceil+k^2\leq p \lceil\frac{|\hat
V|}{p}\rceil-\lceil\frac{|\hat V|}{p}\rceil + 1\leq|\hat V|-1$
vertices of $\hat V$ (and so $X'$ cannot dominate $\hat V$), where
the last inequality holds by the choice of $|\hat V|$ and because $p\geq 3$. 
\qed
\end{proof}

As a consequence of the above claim, if $|X|=\lceil\frac{|\hat V|}{p}\rceil$, i.e., either $X
\subseteq C$ or $X \subseteq \bar{U}$, then either $X$ does not
dominate $r$ in $G(S)$ or $X$ does not dominate $r'$ in $G(S)$. This implies that,
$$
\gamma(S^{v,S'}_{\neg v},v,2)=
\begin{cases}
0 & \text{if $v=r'$;}\\
p & \text{if $v \in C\cup\bar U$;}\\
p+1 & \text{otherwise.}
\end{cases}
$$
Since $|S^{v,S'}_v|=\gamma(S^{v,S'}_{\neg v},v,2)$ for each player $v\in V$, from
Proposition~\ref{prop: bounds for S_v max game} we have that $v$ is in equilibrium w.r.t. $S^{v,S'}$. Therefore, by Proposition~\ref{prop: sufficient conditions for NE}, 
$G(S)$ is stable.

\qed
\end{proof}

\section{\textsc{SumBD}}
\label{sct:sum}

\subsection{Upper bounds}
For \textsc{SumBD}, we start by giving an upper bound to the PoA
similar to the one obtained for \textsc{MaxBD}. For the remaining of this section we use $D$ to
denote the average bound of every node, namely $D=B/n$.

\begin{theorem}\label{th:UB sum}
The PoA of \textsc{SumBD} is $O(n^{\frac{1}{\lfloor\log_3
D/4\rfloor+2}})$ for $D\ge 3$, and $O\left( \sqrt{n \log {n}}
\right)$, when $2 \le D < 3$.
\end{theorem}
\begin{proof}
Let $G=G(S)$ be a stable graph, and let $\rho=SC(S)/(n-1)$. Remind
that the ball of radius $k$ centered at a node $u$ is defined as
$\beta_k(u)=\{v \mid d_G(u,v) \le k\}$. Moreover, let
$\beta_k=\min_{u \in V} |\beta_k(u)|$. We have the following
\begin{claim}
For any $k \ge 1$, we have $\beta_{3k+2} \ge \min \{n/2+1,\lfloor
\rho \rfloor \beta_k\}$.
\end{claim}
\begin{proof}
Consider the ball $\beta_{3k+2}(u)$ centered at any given node
$u$, and assume that $|\beta_{3k+2}(u)| \le n/2$. Let $T$ be the
maximal set of nodes at distance exactly $2k+2$ from $u$ and
subject to the distance between any pair of nodes in $T$ being at
least $2k+1$. We claim that for every node $v \notin
\beta_{3k+2}(u)$, there is a vertex $t \in T$ with $d_G(t,v)\le
d_G(u,v) -2$. Indeed, consider the node $t'$ in the shortest path
between $v$ and $u$ at distance exactly $2k+2$ from $u$. If $t'
\in T$ the claim trivially holds, otherwise consider the node $t
\in T$ that is closest to $t'$. From the maximality of $T$ we have
that $d_G(t,v) \le d_G(t,t')+d_G(t',v) \le 2k+d_G(u,v)-(2k+2) \le
d_G(u,v)-2$.

Let $H$ be the forest consisting of the following disjoint trees.
For every node $t\in T \cup\{u\}$, let $U_t$ be the nodes that are
closer to $t$ than any other $t' \in T \cup\{u\}$, and let $F_t$
be the subtree of the shortest path tree of $G$ rooted at $t$
spanning $U_t$. As a consequence, since $u$ is within the bound in $G$, it
is easy too see that every vertex $x$ is within the bound in $H \cup
\{(x,t) \mid t \in (T \cup\{u\})\setminus \{x\}\}$. Hence, From
Lemma \ref{lm:aux}, we have that $\rho < |T|+1$ and hence $|T|+1
\ge \lfloor \rho \rfloor$. Now, all the balls centered at nodes in
$T\cup\{u\}$ with radius $k$ are all pairwise disjoint. Then:
$$
|\beta_{3k+2}(u)|\ge |\beta_k(u)|+\sum_{t \in T} |\beta_k(t)| \ge
\lfloor \rho \rfloor \beta_k.
$$
\qed
\end{proof}

Now, observe that $\beta_1 \ge \lfloor \rho \rfloor$. Then, after
using the above claim $x$ times, we obtain
$$
\beta_{2 \, 3^{x}-1} \ge \min \{n/2+1,\lfloor \rho
\rfloor^{x+1}\}.
$$

Let us consider the case $R\ge 3$ first. Let $U$ be a maximal
independent set of $G^{D-1}$. Since $U$ is also a dominating set
of $G^{D-1}$, it holds that $|U| \ge \lfloor \rho \rfloor$. We
consider the $|U|$ balls centered at nodes in $U$ with maximal
radius at most $(D-2)/2$. Since $U$ is an independent set of
$G^{D-1}$, all balls are pairwise disjoint and hence we have $n
\ge |U| \lfloor \rho \rfloor^{\lfloor\log_3 D/4\rfloor + 1} \ge
\lfloor \rho \rfloor^{\lfloor\log_3 D/4\rfloor + 2}$. As a
consequence, we obtain $\lfloor \rho \rfloor \le
n^{\frac{1}{\lfloor\log_3 D/4\rfloor}}$, and the claim follows.

Now assume $2\le D < 3$. To use the same argument used for
\textsc{MaxBD}, it suffices to prove that for any stable graph
$G(S)$ with minimum degree $\delta$, it holds that
$\frac{SC(S)}{n-1} \le \min\{\delta+1, O(\gamma(G^{D-1})) \}$. The
upper bound $\frac{SC(S)}{n-1}=O(\gamma(G^{D-1}))$ can be proved
by using the same arguments used in the proof of Lemma \ref{lm:PoA
le gamma} where we exchange the role of $R$ with $D$. Now, we
prove that $\frac{SC(S)}{n-1} \le \delta +1$. Let $v$ be a node
with degree $\delta$, and let
$N_{G(S)}(v)=\{u_1,\dots,u_{\delta}\}$. Consider a shortest path
tree $T$ of $G(S)$ rooted at $v$. Clearly, $v$ is within the bound
in $T$, and if we define $E_x=\{(x,u_j) \mid 1 \le j \le \delta
\}$ for any $x \neq v$, we have $B_{T+E_x}(x) \le B_{G(S)}(v)\le
B$. Hence, from Lemma \ref{lm:aux}, if follows that
$\mathit{SC}(S)\le |E(T)| + (n-1)\delta \le (\delta+1) (n-1)$.
\qed
\end{proof}

From the above result, it follows that the PoA becomes constant when $D=\Omega(n^{\epsilon})$, for some $\epsilon>0$. We now show how to lower such a threshold to $D=2^{\omega(\sqrt{\log n})} = n^{\omega\left(\frac{1}{\sqrt{\log n}}\right)}$ (and we also improve the upper bound when $D=\omega(1) \cap o(3^{\sqrt{\log n}})$).

\begin{lemma}\label{lm:Bcost>B-n}
Let $G(S)$ be stable and let $v$ be a node such that
$B_{G(S)}(v) \le B-n$, then $\mathit{SC}(S) \le 2(n-1)$.
\end{lemma}
\begin{proof}
Let $T$ be the shortest path of $G$ rooted at $v$. The claim
immediately follows from Lemma \ref{lm:aux} by observing that $v$
is within the bound in $T$ and every other node $u$ is within the bound in
$T+(u,v)$. \qed
\end{proof}

Notice that the above Lemma shows that when a stable graph $G$ has
diameter at most $D-1$ then the social cost of $G$ is at most
twice the optimum. Now, the idea is to provide an upper bound to
the diameter of any stable graph $G$ as function of $\delta$,
where $\delta$ is minimum degree of $G$. Then we combine this
bound with Lemma \ref{lm:Bcost>B-n} in order to get a better upper
bound to PoA for interesting ranges of $D$.

The proof of the following theorem follows the schema of that of
Theorem 9 in \cite{ADH10}.

\begin{theorem}\label{th:diam SUM}
Let $G$ be stable with minimum degree $\delta$. Then the diameter
of $G$ is $2^{O(\sqrt{\log n})}$ if $\delta=2^{O(\sqrt{\log
n})}$, and $O(1)$ otherwise.
\end{theorem}
\begin{proof}
We start by proving two lemmas:

\begin{lemma}
Let $G$ be stable with minimum degree $\delta$. Then either $G$
has diameter at most $2 \log n$ or, for every node $u$, there is a
node $x$ with $d_G(u,x)\le \log n$ such that (i) $x$ is buying
$\delta /c$ edges (for some constant $c>1$), and (ii) the removal
of these edges increases the sum of distances from $x$ by at most
$2n(1+\log n)$.
\end{lemma}
\begin{proof}
Assume that the diameter of $G$ is greater than $2 \log n$ and
consider a node $u$. Let $U_j$ be the set of nodes at distance
exactly $j$ from $u$ and let $n_j=|U_j|$. Moreover, denote by $T$
the shortest path tree of $G$ rooted at node $u$. Let $i$ be the
minimum index such that $n_{i+1}< 2n_i$ ($i$ must exist since the
height of $T$ is greater than $\log n$). Consider the set of edges
$F$ of $G$ having both endpoints in $U_{i-1}\cup U_i \cup U_{i+1}$
and that do not belong to $T$. Then, $|F| \ge \delta n_i /2 -
3n_i$. Moreover, we have that $n_{i-1}+n_i+n_{i+1} \le 1/2 \,
n_i+n_i+2n_i=7/2 n_i$. As a consequence, there is a vertex $x \in
U_{i-1}\cup U_i \cup U_{i+1}$ which is buying at least $\frac{n_i
/2 - 3n_i}{7/2 n_i} \ge \delta /c$ edges of $F$, for some constant
$c>1$. Moreover, when $x$ removes these edges, the distance to any
other node $y$ increases by at most $2(1+\log n)$ because
$d_T(x,y) \le 2(1+\log n)$. The claim follows. \qed
\end{proof}

\begin{lemma}\label{lem:add delta/c' edges}
In any stable graph $G$, there is a constant $c'>1$ the addition
of $\delta/c'$ edges all incident to a node $u$ decreases the sum
of distances from $u$ by at most $5n \log n$.
\end{lemma}
\begin{proof}
If $G$ has diameter at most $2 \log n$, then the claim trivially
holds. Otherwise, let $x$ be the node of the previous Lemma and
let $c'$ be such that $\delta/c' \le \delta/c-1$. Moreover, assume
by contradiction that the sum of distances from $u$ decreases by
more than $5n\log n$ when we add to $G$ the following set of edges
$F=\{(u,v_1),\dots,(u,v_h)\}$, with $h=\delta/c'$. Then, let
$F'=\{(x,v_j) \mid j=1,\dots,h \}$. We argue that $x$ can improves
his cost by saving at least an edge as follows: $x$ deletes its
$\delta/c$ edges and adds $F'$. Indeed, the sum of distances from
$x$ increases by at most $2n(1+\log n) \le 4n \log n$ and
decreases by at least $5n \log n - n\log n$, since for every node
$y$ such that the shortest path in $G+F$ from $u$ to $y$ passes
through $x$, we have that $d_G(u,y)-d_{G+F}(u,y) \le \log n$.
Hence, $x$ is still within the bound in $G+F'$ and is saving at least one
edge: a contradiction. \qed
\end{proof}

Recall that the ball of radius $k$ centered at a node $u$ is
defined as $\beta_k(u)=\{v \mid d_G(u,v) \le k\}$. Moreover, let
$\beta_k=\min_{u \in V} |\beta_k(u)|$. We claim that
\begin{equation}\label{eq:ricorrenza potenziata}
\beta_{4k} \ge \min\{n/2+1, \frac{k \delta}{20 c \log n} \beta_k
\},
\end{equation}

\noindent for some constant $c >1$. To prove that, consider the
ball $\beta_{4k}(u)$ centered at any given node $u$, and assume
that $|\beta_{4k}(u)| \le n/2$. Let $T$ be the maximal set of
nodes at distance exactly $2k+1$ from $u$ and subject to the
distance between any pair of nodes in $T$ being at least $2k+1$.
It is easy to see that, from the maximality of $T$, for every node
$v \notin \beta_{3k}$ there is a node $t \in T$ such that
$d_{G}(v,t)\le d_G(u,v)-k$. We assumed that at least $n/2$ nodes
have distance more than $3k$. This implies that there must be a
set $T'\subseteq T$ of size $\delta/c$ such that at least $n
\delta/2|T|$  such vertices $v$ whose distance is at most
$d(u,v)-k$ from some node in $T'$. If we add $\delta/c$ edges from
$u$ to nodes in $T'$, the sum of distances from $u$ decreases by
at least $(k-1)n/2|T| \ge kn/4|T|$. By Lemma \ref{lem:add
delta/c' edges} this improvement is at most $5n \log n$. As a
consequence we have that $|T| \ge \delta k /(20 c \log n)$.
Moreover, all the balls centered at nodes in $T$ are disjoint, and
this proves (\ref{eq:ricorrenza potenziata}). Now, the claim
follows by solving the recurrence (\ref{eq:ricorrenza
potenziata}).\qed
\end{proof}

By using the above theorem along with Lemma \ref{lm:Bcost>B-n}, and observing that if $G(S)$ is stable and has minimum degree $\delta$, then $\frac{SC(S)}{n-1} \le \delta +1$, as shown in the proof of Theorem \ref{th:UB sum}, we have:

\begin{theorem}\label{th:UB sum2}
The PoA of \textsc{SumBD} is $2^{O(\sqrt{\log n})}$ if
$D=\omega(1)$, and $O(1)$ if $D=2^{\omega(\sqrt{\log n})}$.
\qed\end{theorem}

Then, by combining the results of Theorems \ref{th:UB sum} and \ref{th:UB sum2}, we get the bounds reported in Table \ref{UBsum}.

\subsection{Lower bounds}
We can finally prove the following theorem.

\begin{theorem}\label{th:LB sum}
For any $\epsilon >0$ and for $2n-3\leq  B=o(n^2)$,
the PoA of \textsc{SumBD} is at least $2-\epsilon$.
\end{theorem}
\begin{proof}
To prove the theorem, we use the following scheme. First, for
every integer $k\geq 2$, we provide a family ${\cal G}_k$ of
graphs that are stable when
$B\in\big[\lambda(k,n),\lambda'(k,n)\big)$, where $n$ is the size
of the graph and $\lambda(k,n)$ and $\lambda'(k,n)$ are functions
that depend on $k$ and $n$. We also prove that the social cost of
infinitely many graphs in ${\cal G}_k$ is at least $2-\epsilon$
far from the social cost of an optimum, for every $k=o(n)$. Then,
we show that $\lambda(2,n)\leq 2n-3$, $\lambda(k+1,n)\leq
\lambda'(k,n)$, and
$\lambda(\Omega(n),n)=\Omega(n^2)$.

Family ${\cal G}_k$ contains a graph $G_{k,h}$ for every positive
integer $h$. More precisely, $G_{k,h}$ has $n_{k,h}=(h+1)k$
vertices and $m_{k,h}=2kh$ edges. Therefore, the lower bound of
$2-\epsilon$ for the PoA when $k=o(n)$ follows by choosing $h\geq
\frac{2}{\epsilon}-1$. For the rest of the proof, we assume that
$h$ is an arbitrary, but fixed, positive integer. Moreover, with a
little abuse of notation, we will drop the subscript $h$ from
$G_{k,h}$ and $n_{k,h}$ and we will also drop the parameter
$n_k=n_{k,h}$ as argument of the two functions $\lambda$ and
$\lambda'$.

The graph $G_k$ is a highly symmetric graph consisting of $k$
players $\{u_0,\ldots,u_{k-1}\}$ which buy no edge, and, for every
$i=0,\ldots,k-1$, there are $h$ copies of a player (we denote by
$v_i$ any of such players) each buying exactly two edges: one
towards $u_{i}$ and one towards $u_{i+1\bmod k}$. Observe that
$G_k$ has diameter $k$.

The broadcast cost of each player $v_i$ is exactly $\lambda(k)$
while the broadcast cost of each player $u_j$ is equal to $\bar
\lambda(k)$. It is easy to see that $\lambda(2)=2n_2-4,\bar
\lambda(2)=n_2$.

Moreover, one can observe that for every $k\geq 2$
$$
\lambda(k+1)=\lambda(k)+n_k+
\begin{cases}
h & \text{if $k+1$ is even;}\\
1 & \text{if $k+1$ is odd,}
\end{cases}
$$

\noindent as well as

$$
\bar \lambda(k+1)=\bar \lambda(k)+n_k+
\begin{cases}
1 & \text{if $k+1$ is even;}\\
h & \text{if $k+1$ is odd.}
\end{cases}
$$

As each player $v_i$ owns exactly two edges, the only strategy
$v_i$ has to connect to $G_k-v_i$ with exactly one edge, is that
of connecting either to some $v_{i}'$ or to some $u_j$. Therefore,
a lower bound on the broadcast cost  of player $v_i$ if he uses
only one edge to connect to $G_k-v_i$ is $\lambda'(k)=
\min\{\lambda(k),\bar\lambda(k)\}+n_k-1-k$, as $G_k$ has diameter
$k$. Therefore, we have that $G_k$ is stable for every
$B\in\Big[\max\{\lambda(k),\bar\lambda(k)\},\lambda'(k)\Big)$. In
what follows, we show that
$\max\{\lambda(k),\bar\lambda(k)\}=\lambda(k)$, thus proving that
$B\in\big[\lambda(k),\lambda'(k)\big)$, as well as
$\lambda'(k)=\bar \lambda(k)+n_k-1-k$.

Indeed, for every $k\geq 2$, and using $n_{k+1}=n_k+h+1$,  we
have that
$$
\lambda(k+2)=\lambda(k)+2n_{k+1} \text{\,\,\,\, and \,\,\,\,}
\bar\lambda(k+2)=\bar\lambda(k)+2n_{k+1}.
$$
Furthermore, using the relations $n_{k+1}=n_k+h+1, n_k=(h+1)k$,
and the formulas above, $\lambda(3)=2n_3-3,
\bar\lambda(3)=\frac{5}{3}n_3-1$. Therefore, $\bar \lambda(2)\leq
\lambda(2)$ and $\bar \lambda(3)\leq \lambda(3)$. As a
consequence, for every $k\geq 2$, $\bar
\lambda(k+2)=\bar\lambda(k)+2n_{k+1}\leq
\lambda(k)+2n_{k+1}=\lambda(k+2)$. Therefore
$\max\{\lambda(k),\bar\lambda(k)\}=\lambda(k)$.

To complete the proof, it remains to show that $\lambda(2)\leq
2n_2-3$, $\lambda(k+1)\leq \lambda'(k)$, and
$\lambda(\Omega(n))=\Omega(n^2)$. We already proved that
$\lambda(2)=2n_2-4\leq 2n_2-3$. Moreover, for $n_k=2k$, i.e.,
$k=n_k/2$, we have that $G_k$ is a cycle of $2k$ vertices, and
thus the broadcast cost of any vertex is
$\Omega(k^2)=\Omega(n^2)$. Finally, using induction, and observing
that $\lambda(3)\leq \lambda'(2)$, we can prove that
$\lambda(k+1)\leq \lambda'(k)$. Indeed, if $k+1$ is even, then
\begin{eqnarray*}
\lambda(k+1) &=& \lambda(k)+n_k+h\leq \lambda'(k-1)+n_k+h\\
&=&\bar\lambda(k-1)+n_{k-1}+n_{k}-1-k+h\\
&=&\bar\lambda(k)+n_k-1-k=\lambda'(k),
\end{eqnarray*}
while, if $k+1$ is odd, then
\begin{eqnarray*}
\lambda(k+1) &=& \lambda(k)+n_k+1\leq \lambda'(k-1)+n_k+1\\
&=&\bar\lambda(k-1)+n_{k-1}+n_{k}-1-k+1\\
&=&\bar\lambda(k)+n_k-1-k=\lambda'(k).
\end{eqnarray*}
\qed
\end{proof}

\section{Concluding remarks}\label{sct:conclusions}
\hide{The following fact shows an upper bound to the diameter of a
graph in which all nodes are within the bound.

\begin{fact}\label{fact:<=2D}
Let $G$ be a graph such that all nodes are within the bound in \textsc{SumBD}. Then, the
diameter of $G$ is at most $2D$.
\end{fact}
\begin{proof}
Let us consider a node $v$ with eccentricity $k$, and
let $u$ be a node such that $d_G(v,u)=k$. For every node $x$ with
$d_G(v,x)=i$, we have that $d_G(u,i) \ge k-i$, and hence
$d_G(v,x)+d_G(u,x) \ge k$. Therefore:
$$
2Dn= 2 B \ge \sum_{x\in V}d_G(v,x) + \sum_{x\in V}d_G(u,x) =
\sum_{x\in V}(d_G(v,x)+ d_G(u,x)) \ge \sum_{x\in V} k =nk.
$$
\qed
\end{proof}}

In this paper, we have introduced a new NCG model in which the emphasis is put on the fact that a player might have a strong requirement about its centrality in the resulting network, as it may well happen in decentralized computing (where, for instance, the bound on the maximum distance could be used for synchronizing a distributed algorithm). 
We developed a systematic study on the PoA of the two (uniform) games \textsc{MaxBD} and \textsc{SumBD}, which, however, needs to be continued, since a significant gap between the corresponding lower and upper bounds is still open. In particular, it is worth to notice that finding a better upper bound to the PoA would provide a 
better estimation about how much dense a network in equilibrium can be.

Actually, in an effort of reducing such a gap, we focused on \textsc{MaxBD}, and we observed the following fact: 
Recall that a graph is said to be \emph{self-centered} if every node is a
center of the graph (thus, the eccentricity of every node is
equal to the radius of the graph, which then coincides
with the diameter of the graph). An interesting consequence of Lemma
\ref{lm:PoA le gamma} is that only stable graphs that are
self-centered can be dense, as one can infer from the following

\begin{proposition}
Let $G(S)$ be a NE for \textsc{MaxBD} such that $G(S)$ is not self-centered. Then,
$\mathit{SC}(S) \le 2(n-1)$.
\end{proposition}
\begin{proof}
Let $v$ be a node with minimum eccentricity. It must be that
$\varepsilon_{G(S)}(v) \le R - 1$. Then, $U=\{v\}$ is a dominating
set of $G^{R-1}$, and Lemma \ref{lm:PoA le gamma} implies the
claim.\qed
\end{proof}

\noindent
Thus, to improve the lower bound for \textsc{MaxBD}, one has to look to self-centered graphs. 
Moreover, if one wants to establish a lower bound of $\rho$, then a stable graph of minimum degree $\rho-1$ (from Corollary \ref{lm:PoA le delta}) 
is needed.
Starting from these observations, we investigated the possibility to use small and suitably dense self-centered graphs as \emph{gadgets} to build lower bound instances for increasing values of $R$. To illustrate the process, see Figure \ref{fig:pallone}, where using a self-centered cubic graph of diameter 3 and size 20, we have been able to obtain a lower bound of 3 (it is not very hard to see that the obtained graph is in equilibrium).

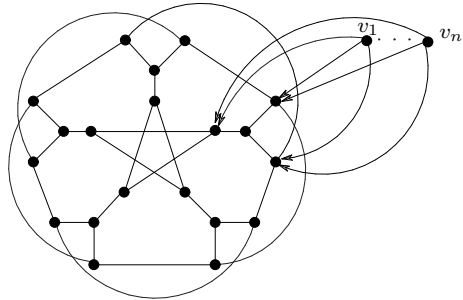
\begin{figure}
    \begin{center}
        \setlength{\unitlength}{0.00066667in}
\begingroup\makeatletter\ifx\SetFigFont\undefined%
\gdef\SetFigFont#1#2#3#4#5{%
  \reset@font\fontsize{#1}{#2pt}%
  \fontfamily{#3}\fontseries{#4}\fontshape{#5}%
  \selectfont}%
\fi\endgroup%
{\renewcommand{\dashlinestretch}{30}
\begin{picture}(3524,2334)(0,-10)
\put(2898,2027){\circle{8}}
\put(3019,2027){\circle{8}}
\put(3135,2027){\circle{8}}
\drawline(1620.000,269.000)(1715.442,279.575)(1808.706,302.446)
	(1898.212,337.227)(1982.448,383.328)(2059.990,439.970)
	(2129.526,506.197)(2189.879,580.886)(2240.030,662.776)
	(2279.130,750.482)(2306.518,842.520)(2321.731,937.334)
	(2324.512,1033.320)(2314.813,1128.855)(2292.798,1222.324)
	(2258.841,1312.146)(2213.514,1396.802)(2157.586,1474.860)
	(2092.000,1545.000)
\drawline(368.000,624.000)(396.343,531.908)(435.479,443.859)
	(484.850,361.114)(543.749,284.856)(611.332,216.177)
	(686.633,156.061)(768.574,105.366)(855.982,64.818)
	(947.607,34.999)(1042.137,16.334)(1138.219,9.091)
	(1234.478,13.373)(1329.538,29.119)(1422.037,56.104)
	(1510.652,93.942)(1594.115,142.090)(1671.230,199.861)
	(1740.896,266.427)(1802.114,340.835)(1854.008,422.021)
	(1895.837,508.824)(1927.001,600.000)
\drawline(203.000,1545.000)(143.658,1470.053)(94.259,1388.209)
	(55.596,1300.781)(28.287,1209.169)(12.770,1114.841)
	(9.293,1019.308)(17.913,924.102)(38.492,830.747)
	(70.699,740.740)(114.018,655.523)(167.756,576.460)
	(231.052,504.820)(302.890,441.751)(382.120,388.261)
	(467.473,345.209)(557.581,313.285)(651.000,293.000)
\drawline(1360.000,2017.000)(1287.962,2080.233)(1208.455,2133.773)
	(1122.775,2176.747)(1032.317,2208.455)(938.557,2228.379)
	(843.023,2236.196)(747.271,2231.777)(652.863,2215.195)
	(561.336,2186.720)(474.184,2146.815)(392.826,2096.133)
	(318.587,2035.498)(252.680,1965.898)(196.176,1888.469)
	(149.999,1804.472)(114.899,1715.276)(91.449,1622.335)
	(80.032,1527.164)(80.833,1431.313)(93.839,1336.346)
	(118.839,1243.810)(155.424,1155.213)(203.000,1072.000)
\drawline(2777.000,2041.000)(2682.398,2047.057)(2587.628,2044.889)
	(2493.402,2034.511)(2400.431,2016.002)(2309.416,1989.501)
	(2221.040,1955.208)(2135.972,1913.382)(2054.850,1864.336)
	(1978.287,1808.442)(1906.858,1746.119)(1841.103,1677.837)
	(1781.516,1604.111)(1728.546,1525.496)(1682.592,1442.584)
	(1644.000,1356.000)
\blacken\drawline(1657.876,1437.286)(1644.000,1356.000)(1694.496,1421.193)(1666.530,1407.268)(1657.876,1437.286)
\drawline(2092.000,1072.000)(2148.234,1149.683)(2194.187,1233.858)
	(2229.119,1323.170)(2252.468,1416.185)(2263.860,1511.407)
	(2263.112,1607.305)(2250.234,1702.337)(2225.435,1794.976)
	(2189.113,1883.732)(2141.852,1967.179)(2084.412,2043.975)
	(2017.716,2112.885)(1942.836,2172.801)(1860.976,2222.761)
	(1773.453,2261.961)(1681.673,2289.771)(1587.112,2305.743)
	(1491.289,2309.622)(1395.746,2301.344)(1302.019,2281.043)
	(1211.614,2249.045)(1125.984,2205.865)(1046.507,2152.196)
	(974.459,2088.902)(911.000,2017.000)
\drawline(2140.000,1096.000)(2233.456,1108.556)(2324.563,1132.867)
	(2411.850,1168.540)(2493.906,1214.999)(2569.405,1271.493)
	(2637.126,1337.109)(2695.976,1410.786)(2745.003,1491.335)
	(2783.414,1577.452)(2810.590,1667.746)(2826.090,1760.759)
	(2829.665,1854.986)(2821.257,1948.906)(2801.001,2041.000)
\blacken\drawline(2216.859,1125.879)(2140.000,1096.000)(2221.877,1086.195)(2195.557,1103.026)(2216.859,1125.879)
\drawline(2116.000,1049.000)(2207.569,1014.135)(2302.628,990.381)
	(2399.833,978.074)(2497.813,977.388)(2595.181,988.331)
	(2690.564,1010.751)(2782.612,1044.330)(2870.026,1088.593)
	(2951.570,1142.915)(3026.093,1206.529)(3092.542,1278.536)
	(3149.977,1357.918)(3197.588,1443.555)(3234.701,1534.236)
	(3260.793,1628.680)(3275.495,1725.552)(3278.598,1823.485)
	(3270.060,1921.094)(3250.000,2017.000)
\blacken\drawline(2197.765,1038.304)(2116.000,1049.000)(2183.112,1001.084)(2168.107,1028.486)(2197.765,1038.304)
\drawline(3274.000,2041.000)(3187.700,2087.269)(3097.698,2125.844)
	(3004.676,2156.430)(2909.344,2178.796)(2812.426,2192.771)
	(2714.658,2198.250)(2616.785,2195.189)(2519.551,2183.613)
	(2423.695,2163.610)(2329.946,2135.332)(2239.018,2098.994)
	(2151.601,2054.871)(2068.361,2003.301)(1989.930,1944.674)
	(1916.905,1879.438)(1849.841,1808.087)(1789.249,1731.165)
	(1735.588,1649.256)(1689.268,1562.983)(1650.641,1473.003)
	(1619.999,1380.000)
\blacken\drawline(1625.353,1462.288)(1620.000,1380.000)(1663.447,1450.088)(1637.080,1433.332)(1625.353,1462.288)
\put(920,2022){\circle*{78}}
\put(1384,2017){\circle*{78}}
\put(1146,1788){\circle*{78}}
\put(1148,1545){\circle*{78}}
\put(911,836){\circle*{78}}
\put(1384,836){\circle*{78}}
\put(1619,1317){\circle*{78}}
\put(651,1309){\circle*{78}}
\put(1856,1309){\circle*{78}}
\put(2092,1545){\circle*{78}}
\put(2092,1072){\circle*{78}}
\put(1620,269){\circle*{78}}
\put(1927,600){\circle*{78}}
\put(1620,600){\circle*{78}}
\put(675,600){\circle*{78}}
\put(368,600){\circle*{78}}
\put(675,269){\circle*{78}}
\put(439,1309){\circle*{78}}
\put(203,1545){\circle*{78}}
\put(203,1072){\circle*{78}}
\put(2801,2017){\circle*{78}}
\put(3278,2009){\circle*{78}}
\drawline(911,2041)(1148,1781)(1384,2017)
\drawline(1148,1805)(1148,1569)
\drawline(1833,1332)(2092,1072)
\drawline(1148,1569)(911,836)(1620,1309)
	(651,1309)(1384,836)(1148,1569)
\drawline(1384,836)(1620,600)(1927,600)
\drawline(1620,624)(1620,269)
\drawline(911,860)(675,600)(675,246)
\drawline(368,600)(675,600)
\drawline(179,1569)(439,1309)(651,1309)
\drawline(203,1072)(439,1309)
\drawline(1596,1309)(1856,1309)(2092,1545)
\drawline(1384,2017)(2092,1545)
\drawline(2069,1096)(2092,1049)
\drawline(2092,1072)(1927,600)
\drawline(651,269)(1620,269)
\drawline(179,1096)(368,600)
\drawline(179,1545)(911,2017)
\drawline(2801,2041)(2116,1569)
\blacken\drawline(2170.528,1630.861)(2116.000,1569.000)(2193.224,1597.923)(2162.113,1600.774)(2170.528,1630.861)
\drawline(3274,2017)(2140,1545)
\blacken\drawline(2206.172,1594.206)(2140.000,1545.000)(2221.543,1557.277)(2191.700,1566.519)(2206.172,1594.206)
\put(3368,2017){\makebox(0,0)[lb]{{\SetFigFont{8}{9.6}{\rmdefault}{\mddefault}{\updefault}$v_n$}}}
\put(2730,2065){\makebox(0,0)[lb]{{\SetFigFont{8}{9.6}{\rmdefault}{\mddefault}{\updefault}$v_1$}}}
\end{picture}
}
    \end{center}
        \caption{A graph with $n+20$ nodes and $3n+30$ edges, showing a lower bound for the PoA of \textsc{MaxBD} for $R=3$ approaching to 3, as soon as $n$ grows. Edges within the gadget (on the left side) are bought by either of the incident nodes, while other edges are bought from the nodes they exit from.}
\label{fig:pallone}
\end{figure}

Interestingly enough, the gadget is a famous extremal (i.e., maximal w.r.t. node addition) graph arising from the study of the \emph{degree-diameter} problem, namely the problem of
finding a largest size graph having a fixed maximum degree and diameter (for a comprehensive overview of the problem, we refer the reader to \cite{degreediam}). More precisely, the gadget is a graph of largest possible size having maximum degree $\Delta=3$ and diameter $R=3$.
In fact, this seems not to be coincidental, since also \emph{Moore graphs} (which are extremal graphs for $R=2$ and $\Delta=2,3,7,57$), and the extremal graph for $R=4$ and $\Delta=3$ (see  \cite{degreediam}), can be shown to be in equilibrium, and then they can be used as gadgets (clearly, the lower bounds implied by Moore graphs for $R=2$ are subsumed by our result in Theorem \ref{th:PoA R=2}). Notice that from this, it follows that we actually have a lower bound of $3$ for the PoA of \textsc{MaxBD} also for $R=4$.
So, apparently there could be some strong connection between the equilibria for \textsc{MaxBD} and the extremal graphs w.r.t. to the degree-diameter problem, and we plan in the near future to explore such intriguing issue.


\end{document}